\documentclass[twocolumn]{autart}    

\usepackage{graphicx}          

\usepackage{amsmath} 
\usepackage{amssymb}  
\usepackage{color}
\usepackage[noadjust]{cite}

\usepackage[switch]{lineno}
\newcommand*\patchAmsMathEnvironmentForLineno[1]{%
	\expandafter\let\csname old#1\expandafter\endcsname\csname #1\endcsname
	\expandafter\let\csname oldend#1\expandafter\endcsname\csname end#1\endcsname
	\renewenvironment{#1}%
	{\linenomath\csname old#1\endcsname}%
	{\csname oldend#1\endcsname\endlinenomath}}%
\newcommand*\patchBothAmsMathEnvironmentsForLineno[1]{%
	\patchAmsMathEnvironmentForLineno{#1}%
	\patchAmsMathEnvironmentForLineno{#1*}}%
\AtBeginDocument{%
	\patchBothAmsMathEnvironmentsForLineno{equation}%
	\patchBothAmsMathEnvironmentsForLineno{align}%
	\patchBothAmsMathEnvironmentsForLineno{flalign}%
	\patchBothAmsMathEnvironmentsForLineno{alignat}%
	\patchBothAmsMathEnvironmentsForLineno{gather}%
	\patchBothAmsMathEnvironmentsForLineno{multline}%
}

\begin{document}

\begin{frontmatter}

\title{A Jointly Optimal Design of Control and Scheduling in Networked Systems under Denial-of-Service Attacks\thanksref{footnoteinfo}} 

\thanks[footnoteinfo]{This paper was not presented at any IFAC 
meeting. Corresponding author Jingyi Lu.}

\author[Germany]{Jingyi Lu}\ead{jingyi.lu@upb.de},    
\author[Australia]{Daniel E.\ Quevedo}\ead{daniel.quevedo@qut.edu.au}              

\address[Germany]{Department of Mechanical Engineering,
	University of Paderborn, Germany.}  
\address[Australia]{School of Electrical Engineering \& Robotics, Queensland University of Technology, Brisbane, Australia.}             


\begin{abstract}                          

We consider the joint design of control and scheduling under stochastic Denial-of-Service (DoS) attacks in the context of networked control systems. A sensor takes measurements of the system output and forwards its dynamic state estimates to a remote controller over a packet-dropping link. The controller determines the optimal control law for the process using the estimates it receives. An attacker aims at degrading the control performance by increasing the packet-dropout rate with a DoS attack towards the sensor-controller channel. We assume both the controller and the attacker are rational in a game-theoretic sense and establish a partially observable stochastic game to derive the optimal joint design of scheduling and control. Using dynamic programming we prove that the control and scheduling policies can be designed separately without sacrificing optimality, making the problem equivalent to a complete information game. We employ   Nash Q-learning to solve the problem and prove that the solution is guaranteed to constitute an $\epsilon$-Nash equilibrium. Numerical examples are provided to illustrate the tradeoffs between control performance and communication cost. 
\end{abstract}

\end{frontmatter}

\section{Introduction}
Cyber-physical systems (CPSs) have successfully integrated computing devices and networking infrastructure to remotely sense and control the physical world. This enables the development of exciting applications in manufacturing, transportation, and smart grid \cite{kim2012cyber}. However, the wireless communication among components of CPSs introduces vulnerabilities against malicious adversaries \cite{chong2019tutorial}. For example, data integrity and availability can be easily deteriorated by replay attacks \cite{mo2009secure} and denial-of-service (DoS) attacks \cite{zhang2015optimal}. This weakness brings a great threat to the security and safety of CPSs, especially when operating in closed-loop. This raises the issue of how to design systematic prevention mechanisms to enhance system security. 

Game theory is often employed as a tool for resilient design and analysis in the presence of adversaries \cite{manshaei2013game}. Assuming the rationality of the adversary,  game-theoretic frameworks can model the interactive decision-making process and yield non-conservative strategies. They have been widely applied to the secure design of control and estimation. For example, Zhu and  Ba\c{s}ar \cite{zhu2015game}  formulated an infinite time horizon zero-sum game to derive a secure control policy that is resilient against various types of cyberattacks. Li \emph{et al.}\ \cite{li2016sinr} presented a power control strategy for remote state estimation in the presence of  DoS attacks through a stochastic game.  The idea has been recently extended to address linear quadratic Gaussian control problems in \cite{zhang2020stochastic}. Miao \emph{et al.}\ \cite{miao2018hybrid} proposed a stochastic game with hybrid states for designing optimal switching control policies to counteract data injection attacks. 

\begin{figure}[t]
	\centering
	\includegraphics[width=0.8\linewidth]{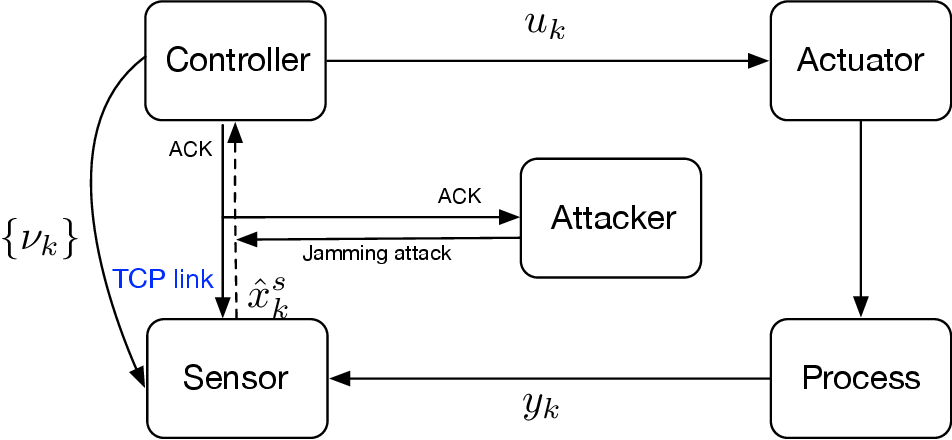}
	\caption{\scriptsize Networked control system with DoS attacks: a local sensor takes measurements of the process outputs $y_k$, generates a state estimate $\hat{x}_k^s$ and transmits the estimate to the remote controller over a stochastic channel according to the scheduling command $\nu_k$ \textcolor{black}{sent from the controller to the sensor via a reliable link}.  After receiving the packet, the TCP link sends an acknowledgment signal (ACK). An attacker intends to degrade the control performance by using information from the acknowledgment signal to selectively jam the communication channel between the controller and the sensor and thereby increasing the packet dropout probability.}
	\label{fig:figcss}
\end{figure}
In this work, we employ the game-theoretic framework to study a secure co-design of control and scheduling to save communication costs and enhance CPS security.
As shown in Fig.\ \ref{fig:figcss}, we consider a remote control problem where the system state estimate is sent from the sensor to the controller across a TCP link affected by random packet dropouts. The communication channel suffers from potential DoS jamming attacks launched by an adversary. The attacker intends to degrade the control performance and increase the controller's communication costs. Since the controller and the adversary have opposite interests, we model this interactive decision-making process by formulating a zero-sum stochastic game with the controller and the adversary taken as two antagonistic players.

To date, control-scheduling co-design has been extensively studied without considering security issues and proved to be an efficient approach to reduce communication cost in networked control systems \cite{cooper1971optimal,antunes2012dynamic,peters2016controller,kogel2019combined,wildhagen2020scheduling}. \textcolor{black}{However, due to the coupling effects of estimation, control and scheduling, co-design methods commonly lead to optimization problems of increased computational complexity, and therefore fail to provide efficient optimal designs.} Similar issues arise when the co-design is conducted in a game-theoretic framework. Specifically, to derive the jointly optimal strategy for the control and scheduling, both the system input and scheduling command have to be taken as actors in the stochastic game. Since the controller output is commonly a continuous-valued variable and the scheduling command is a discrete variable, the resulting game constitutes a game with hybrid states which are difficult to solve \cite{miao2018hybrid}. An additional difficulty arises in the "natural" situation where the cost function is related to the system state  $x_k$, which is randomly available to the controller but  unobserved to the attacker, leading to a  partially observable stochastic game. 

To tackle the challenges in computation, in the present work we explore the possibility of separating the control and scheduling. With the help of dynamic programming, we establish that when the controller works as the scheduler, the optimal feedback gain is independent of packet arrival distributions and realizations; conversely, the optimal scheduling command is independent of the system state estimates. Based on these properties, we can simplify the partially observable stochastic game with hybrid states into a complete information stochastic game with countably infinite discrete states. \textcolor{black}{We further approximate this infinite-state average cost game with a finite-state discounted cost game to enable the application of Nash Q-learning. We prove that the resulting solution provides an $\epsilon$-equilibrium for the original stochastic game. } 
\par We emphasize that, in contrast to existing co-design works, including \cite{zhang2020stochastic,knorn2017optimal,leong2017event}, where the scheduling variable is \emph{a priori} assumed to be independent of the system state, in the present work we do not impose any assumptions on the structure of the solution. Therefore,  the current separate design procedure does not sacrifice optimality. 

The contribution of this work is three-fold: (1) we establish a partially observable stochastic game with hybrid states to derive the jointly optimal scheduling and control strategy under a DoS attack in the context of remote control with random packet dropouts; (2) we prove that the design of control and scheduling can be separated without sacrificing optimality and show that the game can be reduced to a complete information stochastic game with countably infinite discrete states; (3) we show that Nash Q-learning can efficiently solve such problems and guarantee an $\epsilon$-equilibrium.

The remaining parts of this manuscript are organized as follows: Section \ref{sec2} describes the model of the networked control system and the problem setup; Section \ref{sec3} formulates the partially observable stochastic game and analyzes the optimality of the separate design; Section \ref{sec6} presents a truncated policy derived from Nash Q-learning and proves the $\epsilon$-equilibrium of the resulting solution; Section \ref{sec4} provides a numerical example to verify the theoretical results and illustrate the tradeoffs between the communication cost and the control performance; Section \ref{sec5} draws conclusions.

\section{PROBLEM FORMULATION}
\label{sec2}
\subsection{System model}
Consider a linear dynamic system
\begin{align}
x_{k+1} = & A x_k + B u_k +w_k,\label{m1}\\
y_k = & C x_k + v_k,\label{m2}
\end{align}
where $x_k\in\mathbb{R}^{n_x}$, $u_k\in\mathbb{R}^{n_u}$ and $y_k\in\mathbb{R}^{n_y}$ denote the system state, input and output respectively. Without loss of generality, assume $x_0$ is a zero-mean Gaussian noise. $w_k\in\mathbb{R}^{n_x}$ and $v_k\in\mathbb{R}^{n_y}$ are i.i.d.\ Gaussian noises with zero mean, i.e., $\mathbb{E}(w_k)=0$, $\mathbb{E}(v_k)=0$, and $\mathbb{E}(w_iw_j^\top)=\mathbb{E}(v_iv_j^\top)=0$ if $i\neq j$. Denote $\mathbb{E}(w_kw_k^\top)=Q$ and $\mathbb{E}(v_kv_k^\top)=R$, with $Q>0$ and $R>0$. {\color{black}{Assume that the pair $(A,B)$ is stabilizable}} and $(A,C)$ is detectable. 

{\color{black}{As depicted in Fig. \ref{fig:figcss}, a local sensor takes measurements of the system output $y_k$ and generates a state estimate which is denoted as $\hat{x}_k^s$. This estimate is wrapped as a packet and forwarded to the remote controller over a TCP link.	The transmission commits to the command signal $\nu_k\in\{0,1\}$, which is determined by the controller and transmitted to the sensor via a reliable link with no packet dropouts. The optimization of $\nu_k$ is detailed in Section \ref{sec3}. Assume the TCP link has i.i.d. packet dropouts and the probability of a successful reception is $\lambda$. This set up is commonly seen in many remote control problems \cite{leong2018transmission,leong2018transmission1,farjam2019timer}.}}

Define $\gamma_k^c$ as an indicator variable of successful reception. We define
\begin{align}
\gamma_k^c\triangleq\left\{\begin{array}{ll}
0 & \quad \textcolor{black}{\text{if no signal is received,}} \\
1 & \quad \text{if $\hat{x}_k^s$ is successfully received,}
\end{array}\right. \nonumber
\end{align}
and note that
\begin{equation}
\begin{aligned}
&\mathbb{P}(\gamma_k^c=0\mid \nu_k=0) = 1, \\
&\mathbb{P}(\gamma_k^c=0\mid \nu_k=1) = 1-\lambda.
\end{aligned} \label{prob}
\end{equation}
The remote controller determines the input $u_k$ and the scheduling command $\nu_k$ based on the  received information. 
If the packet is successfully received, then the TCP-link will send an acknowledgment back to the sensor, from which the sensor can infer the value of $\gamma_k^c$ and compute its state estimates as described next. 

\subsection{State estimation}
Denote the information set at the sensor and at the controller at time $k$ as $\mathcal{I}_k^s$ and $\mathcal{I}_k^c$, respectively. According to the problem setup described above, we have
\begin{equation}
\begin{aligned}
&\mathcal{I}_k^s=\big\{y_0,\dots,y_k,\hat{x}_0^s,\dots,\hat{x}_k^s,\nu_0,\dots,\nu_k,\gamma_0^c,\dots,\gamma_k^c \big\},\\
&\mathcal{I}_k^c=\big\{\gamma_0^c\hat{x}_0^s,\dots,\gamma_k^c\hat{x}_k^s,\nu_0,\dots,\nu_k,\gamma_0^c,\dots,\gamma_k^c \big\}.
\end{aligned} \label{information}
\end{equation}
It is easy to figure out that \begin{equation}\mathcal{I}_k^c\subset \mathcal{I}_k^s. \label{subset}
\end{equation} {\color{black}{Denote the control policy adopted by the controller as $u_k=\pi_k^c(\mathcal{I}_{k}^c)$. If the sensor knows the function $\pi_k^c$ , it can easily infer $u_k$ from $\mathcal{I}_{k}^s$ since $\mathcal{I}_{k}^c$ is a subset of $\mathcal{I}_{k}^s$ and the sensor knows all information to reconstruct $u_k$.  In this case, the separation principle holds \cite{aastrom2012introduction}. The sensor can run a Kalman filter to obtain the optimal state estimate $\hat{x}_k^s$, which follows $\hat{x}_k^s = \mathbb{E}(x_k\mid \mathcal{I}_{k-1}^s,y_k,u_{k-1})$.}} Specifically, we have 
\begin{equation}
\begin{aligned}
\hat{x}_k^s =& A \hat{x}_{k-1}^s + B u_{k-1}+K_k\Big(y_k-C(A \hat{x}_{k-1}^s + B u_{k-1})\Big) \label{xks}\\
\tilde{x}_k =& (I_{n_x}-K_kC)A\tilde{x}_{k-1}+(I-K_kC)w_{k-1}-\textcolor{black}{K_kv_{k}} \\
P_k^s = &(I_{n_x}-K_kC)\Big(A P_{k-1}^sA^\top +Q\Big).
\end{aligned}
\end{equation}
Here $\tilde{x}_k=x_k-\hat{x}_k^s$ denotes the state estimation error at the sensor \cite{welch1995introduction}. $P_k^s$ denotes the estimation error covariance defined as $P_k^s=\mathbb{E}\left(\tilde{x}_k\tilde{x}_k^\top\right) $
and 
\begin{align}
K_k \triangleq (A P_{k-1}^sA^\top +Q)C^\top \Big(C(A P_{k-1}^sA^\top +Q)C^\top +R\Big)^{-1}.\nonumber
\end{align}
 Initialize $P_0^s$ with large values. According to \cite{anderson2012optimal},  
$P_k^s$ converges exponentially fast. Let $\bar{P}$ denote the steady-state error covariance after the Kalman filter reaches the steady state. {\color{black}{In Remark \ref{assumption_control}, we analyze the rationality of assuming the sensor's knowledge towards the control policy.}} 

Taking into account the packet dropouts, the optimal state estimate at the controller, denoted as $\hat{x}_k$, is derived in the following proposition. 
\begin{prop} \label{prop1}
The optimal state estimate 
follows the   recursion
\begin{align}
\hat{x}_k=\left\{\begin{array}{ll}
\hat{x}_k^s & \quad \text{if}~ {\color{black}{\gamma_{k}^c=1,}}\\
\bar{x}_k & \quad \text{if}~  {\color{black}{\gamma_{k}^c=0,}}
\end{array}\right. \label{update_state}
\end{align}
where \begin{equation}\bar{x}_k=A\hat{x}_{k-1}+Bu_{k-1}.\label{bar_x}
\end{equation}
Accordingly, the expected error covariance at the controller, denoted as $P_k$, follows
	\begin{align}
P_k =&\mathbb{E}((x_k-\bar{x}_k)(x_k-\bar{x}_k)^\top\mid \mathcal{I}_k^c) \nonumber\\
=& \left\{
\begin{array}{ll}
P_k^s& \quad \text{if}~  {\color{black}{\gamma_{k}^c=1,}}\\
f(P_{k-1}) & \quad \text{if}~  {\color{black}{\gamma_{k}^c=0,}}
\end{array}
\right.\label{update_error}
\end{align}
where $f(X)=AXA^\top+Q$.
\end{prop}
{\color{black}{
\begin{pf} If $\gamma_k^c=1$, $\hat{x}_k^s$ is successfully received by the controller. Since $\hat{x}_k^s$ is the optimal estimate of $x_k$, it is easy to obtain that 
\begin{align}
\hat{x}_k =& \arg\min_{\hat{x}_k}~ \mathbb{E}\big((x_k-\hat{x}_k)(x_k-\hat{x}_k)^\top\mid \mathcal{I}_k^c,u_{k-1}=\pi_{k-1}^c(\mathcal{I}_{k-1}^c)\big)\nonumber\\
 =& \arg\min_{\hat{x}_k}~ \mathbb{E}\big((x_k-\hat{x}_k)(x_k-\hat{x}_k)^\top\mid \hat{x}_k^s\big)\nonumber\\
=& \hat{x}_k^s.\nonumber
\end{align}
It can be further verified that the corresponding error covariance follows $P_k =P_k^s$ \cite{welch1995introduction}. Next, we detail the proof of the case $\gamma_k^c=0$. Assume that the latest successful reception occurs at time $k_0$ with $k_0<k$. Namely, $\gamma_{k_0}^c=1$ and $\gamma_{k_0+1},\dots,\gamma_{k}$ all equal to $0$. Since the command signal $\nu_k$ is determined by the controller at time $k$ before the transmission takes place. It only contains the information in $\mathcal{I}_{k-1}^c$, rather than $\mathcal{I}_{k-1}^s$ or $\mathcal{I}_{k}^s$. Therefore, the calculation of $\hat{x}_k$ can be simplified as
\begin{align}
\hat{x}_k =& \arg\min_{\hat{x}_k}~ \mathbb{E}\big((x_k-\hat{x}_k)(x_k-\hat{x}_k)^\top\mid \mathcal{I}_k^c\big)\nonumber\\
=& \arg\min_{\hat{x}_k}~ \mathbb{E}\big((x_k-\hat{x}_k)(x_k-\hat{x}_k)^\top\mid \hat{x}_{k_0}^s,u_{k_0},\dots,u_{k-1}\big)\nonumber\\
=& A^{k-k_0} \hat{x}_{k_0}^s + \sum_{i=1}^{k-k_0} A^{k-1}Bu_{k-i}. \label{pre}
\end{align}
(\ref{pre}) can be written into a recursive form as given in (\ref{update_state}). Accordingly, the evolution of the expected error covariance $P_k$ can be derived and provided in (\ref{update_error}). 
\end{pf}
}}

 We initialize $P_0$ as $P_0^s$. Define $\tau_k$ as the holding time, which indicates the time steps since the last successful transmission, i.e. 
\begin{align}
\tau_k=\left\{
\begin{array}{ll}
0 & \quad \text{if}~ {\color{black}{\gamma_{k}^c=1,}}\\
\tau_{k-1} + 1 & \quad \text{if}~ {\color{black}{\gamma_{k}^c=0,}}
\end{array}
.\right. \label{hold}
\end{align}
From (\ref{update_error}) and 
(\ref{hold}), we can express $P_k$ in terms of $\tau_k$ as $P_k=f^{\tau_k}(P_{k-\tau_k}^s)$,  which can be further simplified as $P_k=f^{\tau_k}(\bar{P})$ after the Kalman filter has reached the steady state. 
Our subsequent analysis  uses the following  properties of the state estimate and the estimation error. 
\begin{lem} \label{lemma1}
	The following facts are true:
	\begin{itemize}
		\item[(a)] $\mathbb{E}[(x_k-\bar{x}_k)\bar{x}_k^\top\mid \bar{x}_k]=0$;
		\item[(b)] $\mathbb{E}[(\hat{x}^s_k-\bar{x}_k)\bar{x}_k^\top\mid \bar{x}_k]=0$;
		\item[(c)] $\mathbb{E}\left((\hat{x}_k^s-\bar{x}_k)(\hat{x}_k^s-\bar{x}_k)^\top \mid \bar{x}_k\right)=  f(P_{k-1})-P_k^s$. 
	\end{itemize} \label{lem}
\end{lem}
\begin{pf}
See Appendix A. 
\end{pf}

\subsection{Model of the adversary}
\label{sec23}
Consider an adversary that intends to degrade the control performance by deteriorating the availability of the communication channel. In particular, the adversary may launch a DoS attack, e.g.,  by increasing the noise at the receiver, and decreasing the successful transmission probability \cite{li2016sinr}.
We model the DoS attack as a binary process where the successful transmission probability is lowered to $\lambda_a<\lambda$ when an attack is launched and remains unchanged otherwise.\footnote{The conclusions presented in our work can be easily extended to adversaries with multi-level attack capabilities.} 

Let $a_k$ be an indicator variable of the DoS attack, i.e. 
\begin{align}
a_k=\left\{
\begin{array}{ll}
1 & \quad \text{if an attack is launched at time $k$},\\
0 & \quad \text{if no attack is launched at time $k$}.
\end{array}
\right.\nonumber
\end{align}
In view of (\ref{prob}), we have 
\begin{equation}
\begin{aligned}
&\mathbb{P}\Big(\gamma_k^c=0\mid (\nu_k,a_k)=(0,j)\Big) = 1,~j\in\{0,1\}, \\
&\mathbb{P}\Big(\gamma_k^c=0\mid (\nu_k,a_k)=(1,0)\Big) = 1-\lambda, \\
&\mathbb{P}\Big(\gamma_k^c=0\mid (\nu_k,a_k)=(1,1)\Big) = 1-\lambda_a.
\end{aligned} \label{attack}
\end{equation}
From (\ref{prob}), (\ref{update_error}) and (\ref{attack}), the evolution of the error covariance $P_k$ in the presence of an attacker can be modeled as a Markov decision process (MDP) where $P_k$ is taken as the state, and the pair $(\nu_k,a_k)$ is  taken as the action. State transition probabilities are computed accordingly as
\begin{equation}
\begin{aligned}
&\mathbb{P}\Big(P_{k}=f(P_{k-1})\mid P_{k-1},(\nu_k,a_k)=(0,j)\Big) = 1, ~j\in\{0,1\},\\
&\mathbb{P}\Big(P_{k}=f(P_{k-1})\mid P_{k-1},(\nu_k,a_k)=(1,0)\Big) = 1-\lambda,\\
&\mathbb{P}\Big(P_{k}=f(P_{k-1})\mid P_{k-1},(\nu_k,a_k)=(1,1)\Big) = 1-\lambda_a,\\ &\mathbb{P}\Big(P_{k}=P_k^s\mid P_{k-1},(\nu_k,a_k)=(1,0)\Big) = \lambda,\\ &\mathbb{P}\Big(P_{k}=P_k^s\mid P_{k-1},(\nu_k,a_k)= (1,1)\Big) = \lambda_a. 
\end{aligned}\label{tran}
\end{equation}
In addition, assume that the attacker can overhear the acknowledgments sent from the controller to the sensor and make use of them to efficiently schedule the attacks. Denote the information available to the attacker at time $k$ as $\mathcal{I}_{k}^a$. We then have
\begin{equation}
\mathcal{I}_k^a=\{\gamma_0,\dots,\gamma_k\}\label{in_attack}
\end{equation} 
and, when comparing with \eqref{information}, note that the attacker has significantly less information than the controller-scheduler.
\section{Co-design of control and scheduling } \label{sec3}
In this section, we formulate the interactive decision making process as a partially observable stochastic game   and show that, from the controller's perspective, the design of the control and scheduling can be conducted separately without sacrificing   optimality. This property simplifies the game into a complete information stochastic game.   

\subsection{Design of the partially observable stochastic game}
Consider that the legitimate controller aims at optimizing the control performance by using the information in $\mathcal{I}_{k-1}^c$ to determine the transmission command $\nu_k$ and adjust the input $u_k$. In particular, the controller intends to minimize a linear combination of a quadratic term of the state $x_k$ and the input $u_k$, as well as the controller's transmission cost and the negative of the adversary's attacking cost:
\begin{equation}
\begin{aligned}J=\frac{1}{K}\mathbb{E}\Big(&\sum_{k=0}^{K-1}x_{k}^\top W x_{k}+u_k^\top U u_k+c_s \nu_k-c_a a_k\\
&+x_K^\top W x_K\mid \mathcal{I}_{k-1}^c\Big).
\end{aligned}  \label{cost}
\end{equation}
Here $c_s$ denotes the cost of transmission \textcolor{black}{generated by the controller and the sensor}. $c_a$ denotes the cost of launching a jamming attack. $W$ and $U$ are positive definite matrices. Conversely, the adversary intends to maximize $J$ by smartly scheduling the attack $a_k$ according to its observations $\mathcal{I}_{k-1}^a$. 

\par Assume the two agents select their actions independently. The interaction can then be formulated as a zero-sum stochastic game. Specifically, we denote the stochastic game $\mathcal{G}$ as a tuple $\langle \mathcal{I},\mathcal{S},\{b^0\},\mathcal{A}_c,\mathcal{A}_a,\mathcal{O}_c,\mathcal{O}_a,\mathcal{P},\mathcal{C}\rangle$, where
\begin{itemize}
	\item $\mathcal{I}$ denotes the number of agents. Here we have two agents, i.e. the controller and the attacker.
	\item $\mathcal{S}$ denotes the set of states $s_k$. Here we define $$s_k=[x_k,P_{k-1}],$$ where $x_k$ is the system state vector taking values from the continuous space $\mathbb{R}^{n_x}$ and the error covariance $P_{k-1}$ takes values from a countable set \begin{equation}\mathcal{I}_p=\{P_k^s,f(P_k^s),\dots,f^{K-k}(P_k^s),\dots\},\label{ds}\end{equation} with $k\in[0,K]$. See (\ref{update_error}).
	\item $\{b^0\}$ denotes the initial value of the state. 
	\item $\mathcal{A}_c$ denotes the action space of the controller. In our problem, the controller has two actors, i.e., the system input $u_k$ which is continuous and the scheduling command $\nu_k\in\{0,1\}$.
	\item $\mathcal{A}_a$ denotes the action space of the attacker, which is $a_k\in\{0,1\}$.
	\item $\mathcal{O}_c$ denotes the set of observations available to the controller. In accordance with the causality, the controller determines the transmission command first and update the system input $u_k$ based on the latest state estimate $\hat{x}_k$ if the packet is successfully received. Therefore, the information available to the determination of $\nu_k$ should be $I_{k-1}^c$, while the information available to the determination of $u_k$ is $I_k^c$.  According to Proposition \ref{prop1}, the optimal estimate of $x_k$ from $\mathcal{I}_{k-1}^c$ is $\bar{x}_k$ and the error covariance $P_{k-1}$ can be inferred from the sequence of $\gamma_k^c$. We define the observation of the actor $\nu_k$ as 
\begin{equation}
o_k^{\nu} = \{ \bar{x}_k,~P_{k-1}\}.\nonumber
\end{equation}	
Given $\mathcal{I}_k^c$, $\hat{x}_k^s$ is available to the controller if $\gamma_k^c=1$. Therefore, we define the observation of the actor $u_k$ as
\begin{equation}
o_k^{u} = \{ \bar{x}_k,~\gamma_k^c\hat{x}_k^s,~P_{k-1}\}.\nonumber
\end{equation}	
In addition, the error covariance $P_k^s$ is taken as common knowledge to all agents.

	\item $\mathcal{O}_a$ denotes the set of observations available to the attacker. As detailed in Section \ref{sec23}, the set $\mathcal{I}_{k-1}^a$ collects all the information available to the adversary, from which only $P_{k-1}$ can be inferred according to (\ref{update_error}). Therefore, we define the observation of the adversary as
 $$o_k^a = \{P_{k-1}\}.$$
	\item $\mathcal{P}$ is the transition probabilities. From (\ref{m1}), we may notice that the transition probability of $x_k$ depends on $u_{k-1}$ and the distribution of $w_{k-1}$:
	\begin{equation}
	\mathbb{P}(x_{k}\mid x_{k-1},u_{k-1})\sim\mathcal{N}(Ax_{k-1}+Bu_{k-1},Q).\nonumber
	\end{equation}
	In addition, the evolution of  $P_k$ depends on the actions $\nu_k$ and $a_k$ as given in (\ref{tran}).
	\item $\mathcal{C}$ denotes the cost function. From (\ref{cost}), we have the immediate cost at time $k$ given as
	\begin{align}
	c_k = x_{k}^\top W x_{k} + u_k^\top U u_k + c_s \nu_k - c_a a_k. \nonumber
	\end{align}
	The controller aims at minimizing the cost while the attacker aims at maximizing the cost. 
\end{itemize}
\begin{rem} \label{rem_game}
Note that the controller can infer $\bar{x}_k$ and $P_{k-1}$ from $\mathcal{I}_{k-1}^c$ according to  (\ref{information}) and (\ref{bar_x}). After the transmission action $\nu_{k}$ is determined, the controller has a stochastic access to $\hat{x}_k^s$ such that $\hat{x}_k$ is available. Beyond that, the attacker can infer $P_{k-1}$ from $\mathcal{I}_{k-1}^a$. Since neither of the two agents have complete information of the state $s_k$, the game is a partially observable stochastic game with a hybrid of discrete and continuous states \cite{miao2018hybrid}. 
In the following section, we employ dynamic programming to show that the design of control and scheduling can be conducted separately without sacrificing optimality. More importantly, this enables us to simplify the problem into a complete information game with countably infinite discrete states that can be efficiently solved.
\end{rem}

\subsection{Separation of scheduling and control}
In this section, we first consider the case where the time horizon $K$ in \eqref{cost} is finite, and then extend the conclusions to infinite horizons, i.e. $K\rightarrow\infty$. To state our first result, we define
\begin{equation}
\begin{aligned}
&S_{K}=W, \quad F_k=(B^\top  S_{k+1}B+U)^{-1}B^\top S_{k+1}A,\\
& S_k= A^\top  S_{k+1}A +W-A^\top S_{k+1}BF_k,
\end{aligned}  \label{siter}
\end{equation}
for $k=K-1,\dots,0$.
\begin{thm} 
	\label{nash}
	For the stochastic game $\mathcal{G}$ with a finite $K$, the stationary Nash equilibrium exists and the design of the control and scheduling can be conducted separately. In particular, the optimal control policy is in the form of 
	\begin{align}	
	u_k^\star = -F_k\hat{x}_k,\label{optimalcontrol}
	\end{align}
	and the optimal transmission policy and the attack policy, denoted as $\pi_k^c$ and $\pi_k^a$ respectively, can be derived by solving the following minimax optimization problem, 
	\begin{align}
	&\min_{\{\nu_k=\pi_{k}^c(\mathcal{I}_{k-1}^c)\}}\max_{\{a_k=\pi_k^a(\mathcal{I}_{k-1}^a)\}} \sum_{k=0}^{K-1}\Big(\text{tr}\big((A^\top  S_{k+1}A+W-S_k)\nonumber\\&\times
	\mathbb{E}(P_k)\big)+ c_s \nu_k - c_a a_k\big)\Big). \label{simple}
	\end{align}
Moreover, the resulting optimal policies depend on the triplet $(P_{k-1},P_k^s,k)$,  such that the  game $\mathcal{G}$ is equivalent to a complete-information stochastic game with finite discrete states.
\end{thm}
{\color{black}{
\begin{pf}
 From the controller's perspective, the game can be formulated in a recursive form as
\begin{equation}
\begin{aligned}
J_{K}^c = &x_{K}^\top S_{K} x_{K},\\
J_k^c= & \min_{\pi_k^c(o_{k}^{\nu}),~\pi_k^u(o_{k}^{u})}~ \mathbb{E}\left(\max_{\pi_k^a(o_{k}^a)} \mathbb{E}\big( c_k+ J_{k+1}^c \big)\right),
\end{aligned} \label{recursive_control}
\end{equation}
with $k=K-1,\dots,0$. Similarly, from the adversary's perspective, the game can be expressed as
\begin{equation}
\begin{aligned}
J_{K}^a = &x_{K}^\top S_{K} x_{K},\\
J_k^a= & \max_{\pi_k^a(o_{k}^a)}\mathbb{E}\left(\min_{\pi_k^c(o_{k}^{\nu}),~\pi_k^u(o_{k}^{u})}  \mathbb{E}\big( c_k+ J_{k+1}^a \big)\right).
\end{aligned}
\end{equation}
First, set $k=K-1$. Denote $\pi_k^{a\star}(o_k^a)$ as the optimal policy adopted by the adversary. (\ref{recursive_control}) can be equivalently expressed as
\begin{align}
J_{k}^c =& \min_{\nu_k}\mathbb{E} \Big(\min_{u_k}\mathbb{E}\Big(\mathbb{E}\big(x_k^\top W x_k + u_k^\top U u_k + c_s\nu_k-c_aa_k \nonumber\\
&+x_{k+1}^\top S_{k+1} x_{k+1}\mid a_k=\pi_k^{a\star}(o_k^a)\big)\mid o_k^{u}\Big)\mid o_k^{\nu}\Big). \label{JKCn1}\\
= &  \min_{\nu_k}\mathbb{E} \Big(\min_{u_k} \mathbb{E}\big(x_k^\top W x_k + u_k^\top U u_k + c_s\nu_k-c_aa_k \nonumber\\
&+x_{k+1}^\top S_{k+1} x_{k+1}\mid  \gamma_k^c\hat{x}_k^s,\bar{x}_k\big)\mid a_k=\pi_k^{a\star}(P_{k-1}),\nonumber\\&\bar{x}_k,P_{k-1}\Big).  \label{JKCn2}
\end{align} 
 (\ref{JKCn2}) holds since the attacker's action $a_k$ only has an impact on $\gamma_k^c$ as shown in (\ref{attack}). 

According to Proposition 1 in \cite{gupta2005lqg}, the optimal input is given as
\begin{equation}
u_k^\star=\left\{\begin{array}{ll}
-F_k\bar{x} & \gamma_k^c=0\\
-F_k\hat{x}_k^s & \gamma_k^c = 1
\end{array}\right. , \label{control_K}
\end{equation}
where the optimal feedback gain is independent of packet arrivals. 

Substitution of (\ref{control_K}) into (\ref{JKCn2}) provides
\begin{equation}
\begin{aligned}
J_k^c=
&\min_{\nu_k}\mathbb{E}\Big( (1-\gamma_k^c)J_{k,\gamma_k^c=0}^c + \gamma_k^c J_{k,\gamma_k^c=1}^c  \mid \bar{x}_k,P_{k-1},\\
&a_k=\pi_k^{a\star}(P_{k-1})\Big),
\end{aligned} \nonumber
\end{equation}
where
\begin{equation}
\begin{aligned}
J_{k,\gamma_k^c=0}^c =
&\bar{x}_k^\top S_k \bar{x}_k + \text{tr}((W+ A^\top S_{k+1} A)P_k\\
&+\text{tr}(S_{k+1}Q)+c_s\nu_k-c_aa_k,\\
J_{k,\gamma_k^c=1}^c
=  & \hat{x}_k^{s\top} S_k \hat{x}_k^s + \text{tr}((W+ A^\top S_{k+1} A)P_k^s)\\
&+\text{tr}(S_{k+1}Q)+c_s\nu_k-c_aa_k.
\end{aligned} \nonumber
\end{equation}
In view of the relationship between $\gamma_k^c$ and $(\nu_k,a_k)$ given in (\ref{attack}), $J_k^c$ is written as
\begin{equation}
\begin{aligned}
J_k^c=
&\min_{\nu_k} \mathbb{E}\Big((1-\nu_k(1-a_k)\lambda-\nu_ka_k\lambda_a)[\bar{x}_k^\top S_k\bar{x}_k+\\
&\text{tr}\big((W+A^\top S_{k+1}A)f(P_{k-1})\big)]+(\nu_k(1-a_k)+\\
&v_ka_k\lambda_a)[\mathbb{E}(\hat{x}_k^{s\top} S_k\hat{x}_k^s\mid \bar{x}_k)+\text{tr}\big((W+A^\top S_{k+1}A)P_k^s\big)]\\
&+c_s\nu_k-c_aa_k\mid a_k=\pi_k^{a\star}(P_{k-1}),P_{k-1}\Big)+\text{tr}(S_{k+1}Q).
\end{aligned} \label{jkc}
\end{equation}
We can further prove with Lemma \ref{lemma1} (b) and (c) that
\begin{equation}
\begin{aligned}
\bar{x}_k^\top S_k \bar{x}_k = & \mathbb{E}\Big(x_k^\top S_k x_k - (x_k-\bar{x}_k)^\top S_k (x_k-\bar{x}_k)\mid \bar{x}_k,P_{k-1}\Big) - \\
& \underbrace{\mathbb{E}\Big(\bar{x}_k^\top S_k(x_k-\bar{x}_k)-(x_k-\bar{x}_k)^\top S_k \bar{x}_k\mid \bar{x}_k,P_{k-1}\Big)}_{=0}\\
=& \mathbb{E}\Big(x_k^\top S_k x_k\mid \bar{x}_k,P_{k-1}\Big) - \mathbb{E}\Big(\text{tr}(S_k(x_k-\bar{x}_k)\times\\
&(x_k-\bar{x}_k)^\top)\Big)\\
=& \mathbb{E}\Big(x_k^\top S_k x_k\mid \bar{x}_k,P_{k-1}\Big)-\text{tr}\big(S_kf(P_{k-1})\big),
\end{aligned} \label{s1}
\end{equation}
and
\begin{equation}
\begin{aligned}
\mathbb{E}(\hat{x}_k^{s\top} S_k \hat{x}_k^s\mid \bar{x}_k) = & \mathbb{E}\Big(\bar{x}_k^\top S_k \bar{x}_k + (\hat{x}_k^s-\bar{x}_k)^\top S_k(\hat{x}_k^s-\bar{x}_k)\mid \bar{x}_k\Big)\\
= & \bar{x}_k^\top S_k \bar{x}_k + \text{tr}\big(S_k(f(P_{k-1})-P_k^s)\big)\\
= & \mathbb{E}\Big(x_k^\top S_k x_k\mid \bar{x}_k,P_{k-1}\Big)-\text{tr}\big(S_kP_k^s\big).
\end{aligned} \label{s2}
\end{equation}
Substitution of (\ref{s1}) and (\ref{s2}) into (\ref{jkc}) yields that 
\begin{equation}
\begin{aligned}
&J_k^c = \mathbb{E}(x_k^\top S_k x_k\mid \bar{x}_k,P_{k-1}) + \text{tr}(S_{k+1}Q)+ \min_{\nu_k}\mathbb{E} \Big[c_s\nu_k-\\
&c_aa_k + \text{tr}\Big((W+A^\top S_{k+1}A-S_k)\big((1-\nu_k(1-a_k)\lambda-\\
&\nu_ka_k\lambda_a)f(P_{k-1})+(\nu_k(1-a_k)\lambda+v_ka_k\lambda_a)P_k^s\big)\Big)\\
&\mid a_k=\pi_k^{a\star}(P_{k-1}),P_{k-1}\Big].
\end{aligned}  \label{jkc_new}
\end{equation}
Grouping all the terms that contain $\nu_k$, we may notice that the coefficients of $\nu_k$ are independent
of the state estimate $\bar{x}_k$. This important property indicates that the optimal scheduling command $v_k$ is independent of $\bar{x}_k$. 
Next, we look at $J_k^a$. Similar to the derivation of $J_k^c$, we have 
\begin{equation}
\begin{aligned}
J_{k}^a =& \max_{a_k}\mathbb{E}\Big( \mathbb{E}\big(x_k^\top W x_k + u_k^\top U u_k + c_s\nu_k-c_aa_k+\\&x_{k+1}^\top S_{k+1} x_{k+1}\mid u_k= \pi_k^{u\star}(o_k^u),\nu_k=\pi_k^{c\star}(o_k^c)\big)\mid o_k^a\Big),\\
=& \mathbb{E}(x_k^\top S_k x_k\mid P_{k-1}) + \text{tr}(S_{k+1}Q)+ \max_{a_k} \mathbb{E}\Big[c_s\nu_k-\\
&c_aa_k + \text{tr}\Big((W+A^\top S_{k+1}A-S_k)\big((1-\nu_k(1-a_k)\lambda\\
&-\nu_k a_k\lambda_a)f(P_{k-1})+(\nu_k(1-a_k)\lambda+v_k a_k\lambda_a)P_k^s\big)\Big)\\
&\mid \nu_k=\pi_k^{\nu\star}(P_{k-1}),P_{k-1}\Big].
\end{aligned} \label{jka}
\end{equation}
 Compare (\ref{jkc_new}) and (\ref{jka}). $J_k^c$ minimizes a convex function and $J_k^a$ maximizes a concave function. According to Kakutani fixed-point theorem \cite{kakutani1941generalization} , a Nash-equilibrium exists. For  $k\in[0,K-1)$, the existence of the Nash-equilibrium can be proved in the same way.

To complete the proof, note that in (\ref{jkc_new}) and (\ref{jka}), the decision variables both depend on $P_{k-1}$ and the matrix $W+A^\top S_{k+1}A$, which is a function of time $k$. Beyond that, when $k\in[0,K-1]$, $P_{k-1}$ takes values from a finite set according to (\ref{xks}) and (\ref{update_error}) and its value is available in both $o_k^\nu$ and $o_k^a$. Hence, the game is a complete-information game with finite discrete states. Since the schedule commands $\nu_k$ and $a_k$ are both independent of the state $\bar{x}_k$, from  \cite[Theorem 1]{leong2017event}, the optimal cost is in the form of
\begin{equation}
\begin{aligned}
& \mathbb{E}(x_0^\top S_0 x_0)+\sum_{k=0}^{K-1}\text{tr}(S_{k+1}Q)+ \sum_{k=0}^{K-1}\Big(\text{tr}\big(( A^\top S_{k+1}A\\
&+W-S_k)\mathbb{E}(P_k)\big)+ c_s \nu_k - c_a a_k \Big).
\end{aligned} \label{long}
\end{equation}
Since the first two terms are irrelevant to $\{\nu_k\}$ and $\{a_k\}$, the solution to (\ref{simple}) is the same as that to (\ref{long}).
\end{pf} }}
\begin{rem}
	Theorem \ref{nash} establishes that in the present setting, the optimal scheduling command $\nu_k$ is independent of the state estimate. This fundamental property is often assumed without detailed justification in co-design works, e.g., \cite{leong2017event}.  
	\par Our result guarantees that the separate design does not sacrifice optimality when the controller works as the scheduler. However, if the sensor schedules, such as considered in \cite{zhang2020stochastic,knorn2017optimal},  then this separation property does not hold. In fact, consider that $\nu_k$ is determined by the sensor, so that  $\hat{x}_k^s$ and $\bar{x}_k$ are available before the decision is made. $o_k^\nu$ contains $hat{x}_k^s$ and therefore the expectations in  (\ref{JKCn2}) should be computed conditioned on both $\bar{x}_k$ and $\hat{x}_k^s$, in which case both $\bar{x}_k$ and $\hat{x}_k^s$ will appear in the coefficient terms of $\nu_k$. Therefore, the optimal schedule $\nu_k$ will in general depend on $\bar{x}_k$ and $\hat{x}_k^s$. Since $\nu_k$ contains information of $\hat{x}_k^s$, the optimal estimate at the controller does not follow (\ref{update_state}), but becomes a nonlinear estimation problem, as analyzed in \cite{li2019information,wu2016finite}, and will further complicate the design of the feedback gain. \hfill $\square$.
\end{rem}

Theorem \ref{nash} can be easily extended to infinite horizons provided that the value of $J$ is bounded. 
\begin{cor} \label{cor2}
	Suppose that {\color{black}{ \begin{equation}\lambda_a>1-\frac{1}{|\sigma_{\max}(A)|^2},\label{condition}
	\end{equation} 
	where $|\sigma_{\max}(A)|$ denotes the spectral radius of $A$.}}
	Then the infinite horizon stochastic game $\mathcal{G}$ has a Nash equilibrium. The optimal control policy is given as 
	\begin{align}	
	u_k^\star = - (B^\top  S_{\infty}B+U)^{-1}B^\top S_{\infty}A\hat{x}_k,\label{infinite}
	\end{align}
	where, cf (\ref{siter}), 
	\begin{multline*}
	S_{\infty}= A^\top S_{\infty}A+W
	-A^\top S_{\infty}B( B^\top S_{\infty}B+U)^{-1}B^\top S_{\infty}A.
	\end{multline*}
	The optimal scheduling policies  can be derived from 
	\begin{equation}
	\begin{aligned}
	\min_{\{\nu_k=\pi^c(P_{k-1})\}}&\max_{\{a_k=\pi^a(P_{k-1})\}}\tilde{J}_{(\pi^c,\pi^a)}=\lim_{K\rightarrow\infty}\frac{1}{K} \sum_{k=0}^{K-1}\Big(c_s \nu_k \\&- c_a a_k+\text{tr}\Big(M_{\infty}\mathbb{E}(P_k)\Big)\Big),\label{obja}
	\end{aligned}
	\end{equation}
	where $M_{\infty}=A^\top  S_{\infty}A+W-S_{\infty}$ and $P_k$ takes values from the countably infinite set $\{\bar{P},f(\bar{P}),f^2(\bar{P})\dots\}$ cf (\ref{ds}).
\end{cor}
In Corollary \ref{cor2},  (\ref{condition}) gives a sufficient condition to ensure the boundedness of (\ref{obja}), since even if the attacker launches the DoS attack at all time instances, the expected estimation error covariances are ensured to be bounded. 
Given that the Kalman filter converges exponentially fast, we can study the infinite horizon optimization problem in (\ref{obja}) by assuming the steady state of the Kalman filter, i.e. $P_k^s=\bar{P}$. Moreover, since the feedback gain in (\ref{infinite}) also converges to a constant matrix, the scheduling policies become time-independent. They solely depend on $P_{k-1}$.  Therefore, the infinite horizon stochastic game $\mathcal{G}$ is equivalent to a complete-information game with countably infinite states. 
{\color{black}{
\begin{rem} \label{assumption_control}
We may notice that the control policy $\pi_k^c$ is a standard state feedback strategy with time dependent feedback gain (the finite horizon case) or contant feedback gain (the infinite horizon case) according to Theorem \ref{nash} and Corollary \ref{cor2}. The gain is  independent of the packet dropouts $\gamma_k$ as well as the attacker's policy. Therefore, the controller can synchronize the control policy with the sensor before the process starts such that the separation principle holds at the sensor.
\end{rem}
}}

\section{Nash Q-learning} 
\label{sec6}
In this section, we detail the algorithms to solve the stochastic game $\mathcal{G}$ in practice. If $K$ is finite, the problem can be easily solved using the idea of dynamic programming. We focus on the case where $K$ is infinite. 

Nash Q-learning is a well known multi-agent reinforcement learning algorithm for stochastic games \cite{hu2003nash,li2016sinr}. However, it is only applicable to MDPs with finite states and discounted costs. To apply Nash Q-learning to the current setup, we have to do two modifications to the original problem in (\ref{obja}): (1) aggregate the states of the MDP to make the states finite; (2) approximate the average cost with a discounted cost. In particular, we consider a stochastic game, denoted as $\mathcal{G}^N$, in the form of  
	\begin{equation}
\begin{aligned}
\hat{J}=\min_{\{\nu_k=\pi^c(\bar{\tau}_{k-1})\}}&\max_{\{a_k=\pi^a(\bar{\tau}_{k-1})\}}\lim_{K\rightarrow\infty}  \sum_{k=0}^{K-1}  \eta^k\mathbb{E}\Big[c_s \nu_k \\& -c_a a_k+\text{tr}\Big(M_\infty f^{\bar{\tau}_k}(\bar{P})\Big)\Big],
\end{aligned}\label{obja_2}
\end{equation}
where the discount factor $0<\eta<1$ and $\bar{\tau}_k$ follows, cf. (\ref{hold}),  
\begin{align}
\bar{\tau}_k=\left\{
\begin{array}{ll}
0 & \quad \nu_k\gamma_k^c=1\\
\max\{\bar{\tau}_{k-1} + 1,N-1\} & \quad \nu_k\gamma_k^c=0
\end{array}
.\right. \label{hold2}
\end{align}
\textcolor{black}{In (\ref{obja_2}), we replace $\mathbb{E}(P_k)$ with $\mathbb{E}\big(f^{\bar{\tau}_k}(\bar{P})\big)$} since there is a one-to-one correspondence between
$P_{k}$ and the holding time $\bar{\tau}_{k}$ according to (\ref{hold}). Compare (\ref{hold2}) and (\ref{hold}). In (\ref{hold2}), the states with $\tau_k\geq N-1$ are aggregated as a single state. In this way, we keep $N$ states in the MDP to facilitate the application of Nash Q-learning. In particular, we learn the $Q$-function, denoted as $Q(\bar{\tau},\nu_k,a_k)$, for $\bar{\tau}\in[0,N-2]$ according to the Nash Q-learning algorithm presented in \cite{li2016sinr}. The transmission and attack policies for $\bar{\tau}\in[0,N-2]$ are derived from (\ref{policy})
\begin{equation}
\begin{aligned}
&\pi_N^{c}(\bar{\tau}) =
\underset{\pi_N^c}{\arg\min}\underset{\pi_N^a}{\max}\underset{\nu_k,a_k}{\sum} Q(\bar{\tau},\nu_k,a_k)\mathbb{P}(\nu_k|\pi_N^c)\mathbb{P}(a_k|\pi_N^a),\\
&\pi_N^{a}(\bar{\tau}) =
\underset{\pi_N^a}{\arg\max}\underset{\pi_N^c}{\min}\underset{\nu_k,a_k}{\sum} Q(\bar{\tau},\nu_k,a_k)\mathbb{P}(\nu_k|\pi_N^c)\mathbb{P}(a_k|\pi_N^a).\\
\end{aligned}\label{policy}
\end{equation} 
Here $\pi_N^c(\bar{\tau})$ and $\pi_N^a(\bar{\tau})$ specify $\mathbb{P}(\nu_k\mid \bar{\tau}_{k-1}=\bar{\tau})$ and $\mathbb{P}(a_k\mid \bar{\tau}_{k-1}=\bar{\tau})$ for each $\bar{\tau}\in[0,N-2]$ and each $\nu_k,a_k\in\{0,1\}$. 
For $\bar{\tau}= N-1$, we set \begin{equation}
\begin{aligned}&\mathbb{P}(\nu_k=1\mid \bar{\tau}_{k-1}=N-1)=1, \\
&\mathbb{P}(a_k=1\mid \bar{\tau}_{k-1}=N-1)=1.
\end{aligned} \label{prob_Q}
\end{equation} 

Next, we show the $\epsilon$-optimality of the truncated policy. 

\begin{thm}
	Denote the truncated policy in (\ref{policy}) as $(\pi_N^{c^\star},\pi_N^{a^\star})$. Then, when $\eta$ is close to $1$ and $N$ is large, the pair $(\pi_N^{c^\star},\pi_N^{a^\star})$ provides an $\epsilon$-Nash equilibrium of the game $\mathcal{G}$, i.e.
	\begin{align}
	&\tilde{J}_{(\pi_N^{c^\star},\pi_N^{a^\star})}\leq \min_{\pi^c}\tilde{J}_{(\pi_c,\pi_N^{a\star})}+\epsilon_{\eta,N},\label{ep1}\\
	&\tilde{J}_{(\pi_N^{c^\star},\pi_N^{a^\star})}\geq\max_{\pi^a} \tilde{J}_{(\pi_N^{c^\star},\pi^{a})}-\epsilon_{\eta,N}, \label{ep2}
	\end{align}
	where $\lim_{\eta\rightarrow 1,~N\rightarrow\infty}\epsilon_{\eta,N}=0$ and $\tilde{J}$ is defined in (\ref{obja}). \label{equilibrium}
\end{thm} 
\begin{pf} See Appendix \ref{pfeb}. 
\end{pf}

Theorem \ref{equilibrium} indicates that the truncated policy derived from the discounted cost can make the value of the game arbitrarily close to the Nash-equilibrium by increasing the truncation horizon $N$ and controlling the discounted factor  $\eta$ close to $1$.

\section{Simulation} \label{sec4}
{\color{black}{Consider an infinite time horizon co-design problem with the system matrices set as
\begin{align}
A = \left[\begin{array}{cc}
1.2 & 1\\ 0 & 1
\end{array}\right],\quad B=\left[\begin{array}{c}
1\\2
\end{array}\right],\quad C = \left[\begin{array}{c}
1\\1
\end{array}\right]^\top. \nonumber
\end{align}
The noise covariances are taken as $Q=I_{n_s}$ and $R=I_{n_u}$. The weight matrices are set as $U=I_{n_u}$ and $W=I_{n_s}$. When there is no attack, the successful transmission rate is $\lambda=0.95$ and the transmission cost is set as $c_s=2000$. Accordingly, the optimal control law is computed from (\ref{infinite}) as $L=\left[\begin{array}{cc} -0.39 & -0.63 \end{array}\right]$.

We study the interaction between the attacker and the controller. Fix the attacking cost as $c_a=2000$. At instances that an attack is launched, the successful transmission rate drops to $\lambda_a$. We vary the parameter $\lambda_a$ from $0.3$ to $0.9$ to simulate attackers with different level of  interfere ability. A small $\lambda_a$ indicates that the attacker is powerful to cause significant packet dropouts. For each $\lambda_a$, the game is solved correspondingly via the Nash-Q learning in \cite{li2016sinr} with the truncation horizon $N=50$ and the discount factor $\eta = 1-(0.1)^{8}$. The resulted transmission policies present a threshold-type structure, i.e. 
\begin{align}
\nu_k = & \left\{
\begin{array}{cc}
0 & \bar{\tau}<\text{thres}_c\\
1 & \bar{\tau}\geq \text{thres}_c
\end{array}
\right. ,\quad
a_k =  \left\{
\begin{array}{cc}
0 & \bar{\tau}<\text{thres}_a\\
1 & \bar{\tau}\geq \text{thres}_a
\end{array}
\right. .\label{thres}
\end{align}
We run a Monte Carlo simulation with the derived transmission and attack policies with the simulation length $T$ taken as $10^6$.
As shown in Fig. \ref{case1}(a),  the averaged cost $\tilde{J}$ defined in (\ref{obja}) monotonically decreases with the increase of $\lambda_a$, indicating that the weaker the attacker is, the more easier the controller achieves a smaller cost. From Fig. \ref{case1} (b), we may note that the attack threshold $\text{thres}_a$ is always greater than or equal to the transmission threshold $\text{thres}_c$ to save the attacking cost. Meanwhile, we calculate the averaged transmission probability $p_c$ and averaged attack probability according to 
\begin{equation}
p_c=\frac{\sum_{k=0}^T \nu_k}{T+1},\quad p_c=\frac{\sum_{k=0}^T a_k}{T+1}.\nonumber
\end{equation}
Their trajectories are plotted in Fig. \ref{case1} (b). Compare Fig. \ref{case1} (b) and (c). When $\lambda_a$ gets larger than $0.45$, the attacker chooses to stay silent in most of the time since it can hardly cause a significant degrade of the controller's tracking performance while paying a high cost on attacking. From the controller's perspective, the transmission probability $p_c$ is decreasing when $\lambda_a\geq0.38$ to save transmission cost when the attacker becomes inactive. 

This numerical experiment shows that both the controller and the attacker achieve the Nash-equilibrium by learning from the interactions between each other without the knowledge of either $\lambda$ or $\lambda_a$. This is the main advantage of the game theoretic approaches. To verify this, we compare with the conventional fix-threshold approach. Assume that the adversary is smart and can adapt its own policy according to the controller's via the standard Q-learning. In Fig. \ref{comparison}, we compare the averaged cost $\tilde{J}$ corresponding to $\text{thres}_c=1,3,6$ with the cost $\tilde{J}$ derived from the Nash Q-learning. We may note that the policy derived with the Nash-Q learning always outperform the fixed threshold policy. 
}}

\begin{figure}[t]
	\centering
	\includegraphics[width=0.8\linewidth]{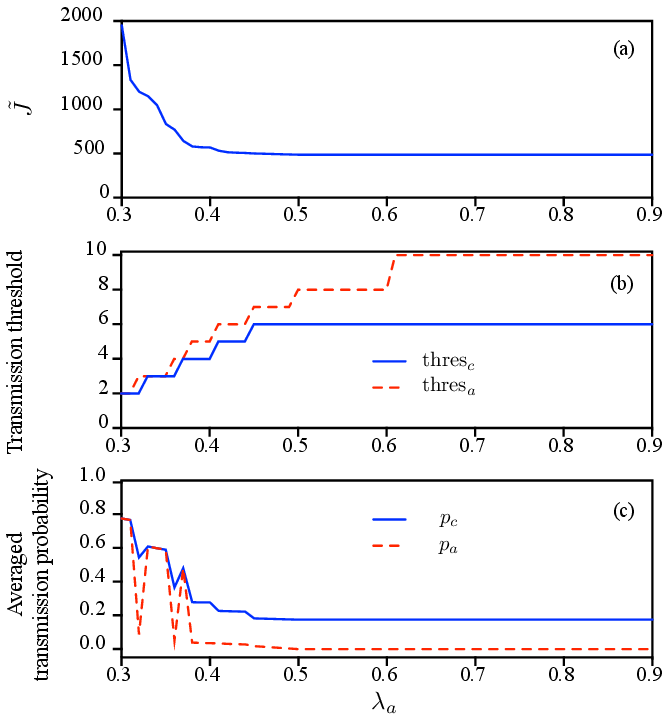}
	\caption{\scriptsize Trajectories of the cost function $\tilde{J}$, the transmission threshold and the averaged transmission probability with $\lambda_a$ varying from $0.3$ to $0.9$.  }
	\label{case1}
\end{figure}

\begin{figure}[b]
	\centering
	\includegraphics[width=1\linewidth]{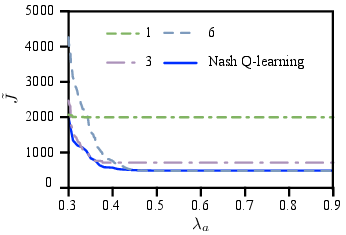}
	\caption{\scriptsize Comparisons between the Nash Q-learning and threshold-type policies with $\text{thres}_c=1,3,6$. }
	\label{comparison}
\end{figure}


\section{Conclusions}
\label{sec5}
In the present work, we have studied the design of control and scheduling in the presence of stochastic  DoS attacks. We have designed  a zero-sum partially observable stochastic game to jointly optimize the control and scheduling strategies. We have established that the optimal transmission scheduling command is independent of the state estimate and that the design of control and scheduling can be conducted separately without sacrificing optimality. Based on this property, the partially observable game is simplified as a complete information game. Moreover, we have proved that by applying a Nash Q-learning to the stochastic game with appropriately truncated states, an $\epsilon$-equilibrium can be derived. Future works may consider employing the sensor as the scheduler, in which case more information can be exploited for scheduling and enhanced security can be expected. 

\appendix

\section{Proof of Lemma \ref{lemma1}}
\label{pf_lemma1}
\begin{pf}
	Item (a) is the same as Lemma 4.1 in \cite{schenato2007foundations}. Item (b) can be proved based on \cite{demirel2018tradeoffs}. For that purpose, define $$r_{k-1}\triangleq  K_kCA\tilde{x}_{k-1}+K_kC w_{k-1}+K_kv_{k}.$$ Assume that the latest successful transmission occurs at time $k_0$ where $k_0\leq k$. According to Eq. (27) in  \cite{demirel2018tradeoffs}, we have
	$$\hat{x}_k^s-\bar{x}_k=\sum\nolimits_{i=k_0}^{k-1}A^{k-1-i}r_i.$$
	Note that $r_i$ is independent of $\hat{x}_{k_0}^s$ for any $ k_0\leq i \leq k-1$, and $\bar{x}_k$ and $u_i$ both depends on $\hat{x}_{k_0}^s$. It can be concluded that $\hat{x}_k^s-\bar{x}_k$ is independent of $\bar{x}_k$ and therefore Item (b) is proved. 
	
	Moreover, from Item (a) and (b), we have $\mathbb{E}(\tilde{x}_k\bar{x}_k^\top\mid \bar{x}_k)=0$, i.e. $\tilde{x}_k$ is independent  of $\bar{x}_k$.
	Hence, we have $
	\mathbb{E}\left(\tilde{x}_k\tilde{x}_k^\top \mid \bar{x}_k\right) =  \mathbb{E}\left(\tilde{x}_k\tilde{x}_k^\top \right)=P_k^s$. Then, Item (c) can be proved as follows:
	\begin{equation}
	\begin{aligned}
	&\mathbb{E}\left((\hat{x}_k^s-\bar{x}_k)(\hat{x}_k^s-\bar{x}_k)^\top \mid \bar{x}_k\right) \\
	= & \mathbb{E}\left((\hat{x}_k^s-x_k+x_k-\bar{x}_k)(\hat{x}_k^s-x_k+x_k-\bar{x}_k)^\top \mid \bar{x}_k\right) \\
	=& \mathbb{E}(\tilde{x}_k\tilde{x}_k^\top+(x_k-\bar{x}_k)(x_k-\bar{x}_k)^\top - 2\tilde{x}_k\underbrace{(x_k-\bar{x}_k)^\top}_{=\tilde{x}_k+\hat{x}_k^s-\bar{x}_k}\mid \bar{x}_k) \\
	=& \mathbb{E}((x_k-\bar{x}_k)(x_k-\bar{x}_k)^\top -\tilde{x}_k\tilde{x}_k^\top\mid \bar{x}_k)=  f(P_{k-1})-P_k^s.
	\end{aligned} \nonumber
	\end{equation}
\end{pf}

\section{Proof of Theorem \ref{equilibrium}} \label{pfeb}
\begin{pf}
	
	Denote $J^{\eta,N}_{(\pi_N^{c^\star},\pi_N^{a^\star})}$ as the value of (\ref{discount}) 
	\begin{equation}
	\begin{aligned}
	\lim_{K\rightarrow\infty}  (1-\eta)&\sum_{k=0}^{K-1} \eta^k\mathbb{E}\Big[c_s \nu_k  -c_a a_k+\text{tr}\Big(M_{\infty}f^{\bar{\tau}}(\bar{P})\Big)\Big],
	\end{aligned} \label{discount}
	\end{equation}
when $(\pi_N^{c^\star},\pi_N^{a^\star})$ is applied. 	
	Since $\pi_{N}^{c^\star}$ and 
	$\pi_{N}^{a^\star}$ give the stationary Nash equilibrium, we have
	\begin{equation}
	\max_{\pi_N^a} J^{\eta,N}_{(\pi_N^{c^\star},\pi_N^{a})}\leq J^{\eta,N}_{(\pi_N^{c^\star},\pi_N^{a^\star})} \leq \min_{\pi_N^c} J^{\eta,N}_{(\pi_N^c,\pi_N^{a^\star})}.
	\end{equation}
Next, we prove  (\ref{ep1}) by showing that there exists a scalar $\epsilon_{\eta,N}$ such that $\lim_{\eta\rightarrow 1,N\rightarrow\infty} \epsilon_{\eta,N}=0$ and 
	\begin{align}
	&\big|J^{\eta,N}_{(\pi_N^{c^\star},\pi_N^{a^\star})}-\min_{\pi^c}\tilde{J}_{(\pi^{c},\pi_N^{a^\star})}\big|<\frac{1}{2}\epsilon_{\eta,N},\label{min_t} \\	&\big|J^{\eta,N}_{(\pi_N^{c^\star},\pi_N^{a^\star})}-\tilde{J}_{(\pi_N^{c^\star},\pi_N^{a^\star})}\big|<\frac{1}{2}\epsilon_{\eta,N} .\label{ep}
	\end{align}
First, fix the attacker's policy to be $\pi_N^{a^\star}$. Define \begin{equation}
\begin{aligned}
&\pi^{c_0}=\arg\min_{\pi^c}\tilde{J}_{(\pi^{c},\pi_N^{a^\star})},\\ &\pi^{c_1}=\arg\min_{\pi^{c}}J^{\eta,\infty}_{(\pi^{c},\pi_N^{a^\star})}.
\end{aligned} \nonumber
\end{equation} According to the boundedness of $\tilde{J}$ and Abel's theorem \cite{hernandez2012discrete}, for any $\eta<1$, there exists a scalar $\tilde{\epsilon}_\eta>0$, such that
\begin{equation}
| \tilde{J}_{\pi^{c_0},\pi_N^{a^\star}}-J^{\eta,\infty}_{\pi^{c_0},\pi_N^{a^\star}}|\leq \tilde{\epsilon}_\eta, \label{com0}
\end{equation} 
and $\lim_{\eta\rightarrow 1}\tilde{\epsilon}_\eta=0$. Eq. (\ref{com0}) shows that the optimality gap induced by the discount factor diminishes when $\eta$ goes to $1$. Following this, we further analyze the optimality gap induced by state truncation.

From the transmission policy $\pi_N^{c^\star}$, we can derive the occupation measures, denoted as $\bar{\omega}\big(\bar{\tau},(\nu,a)\big)$, for the corresponding truncated MDP. Here $\bar{\tau}\in[0,N-1]$ and  $\bar{\omega}$
can be interpreted as the total expected discount time spent in the state-action pairs $\big(\bar{\tau},(\nu,a)\big)$ (Section 3.1 in \cite{altman1999constrained}).  Inspired by the Bellman equation, the occupation measures satisfy that
\begin{equation}
\begin{aligned}
&\sum_{\bar{\tau}\in[0,N-1]}\sum_{\nu,a\in\{0,1\}}\bar{\omega}(\bar{\tau},(\nu,a))=1,\quad \bar{\omega}(\bar{\tau},(\nu,a))\geq 0\\
&\sum_{\bar{\tau}\in[0,N-1]}\sum_{\nu,a\in\{0,1\}}\bar{\omega}(\bar{\tau},(\nu,a))(1\{\bar{\tau}'=\bar{\tau}\}-\eta \mathcal{P}_{\bar{\tau}\{\nu,a\}\bar{\tau}'})\\
&=(1-\eta)1(\bar{\tau}'=0),
\end{aligned}\label{occupation}
\end{equation}
Here $1(\cdot)$ denotes an indicator function. $\mathcal{P}_{\bar{\tau}\{\nu,a\}\bar{\tau}'}$ is a short form for the transition probability $\mathbb{P}\big(s_{k+1}=\bar{\tau}'\mid s_k=\bar{\tau},\nu_k=\nu,a_k=a\big)$. According to (\ref{occupation}) and the transition probability given in (\ref{tran}), we have
\begin{equation}
\begin{aligned}
&\sum_{\nu,a\in\{0,1\}}\bar{\omega}\big(\bar{\tau}+1,(\nu,a)\big) = \eta \sum_{a\in\{0,1\}} \bar{\omega}\big(\bar{\tau},(\nu=0,a)\big)\\
& +\eta \lambda 
\bar{\omega}\big(\bar{\tau},(\nu,a)=(1,1)\big)+\eta \lambda_a\bar{\omega}\big(\bar{\tau},(\nu,a)=(1,0)\big).
\end{aligned} \label{exponential}
\end{equation}
Since $0<\lambda_a<\lambda<1$,  $\bar{\omega}$ exponentially decreases, i.e. $$\sum_{\nu,a\in\{0,1\}}\bar{\omega}\big(\bar{\tau}+1,\nu,a\big)\leq \eta\sum_{\nu,a\in\{0,1\}}\bar{\omega}\big(\bar{\tau},\nu,a\big).$$
In addition, the transmission policy can be derived from $\bar{\omega}$ as 
\begin{equation}
\mathbb{P}\big((\nu_k,a_k)\mid\bar{\tau}\big)=\frac{\bar{\omega}(\bar{\tau},(\nu,a))}{\sum_{\nu_k,a_k\in\{0,1\}}\bar{\omega}(\bar{\tau},(\nu,a))}. \label{prob_occu}
\end{equation}
The value of the cost functions can be expressed in terms of $\bar{\omega}$ as
	\begin{align}
J^{\eta,N}_{(\pi_N^{c^\star},\pi_N^{a^\star})} =& \sum_{\bar{\tau}=0}^{N-1}\sum_{\nu,a\in\{0,1\}} \bar{\omega}\big(\bar{\tau},(\nu,a)\big)\text{tr}f^{\bar{\tau}}\big(\bar{P}\big).\label{trun}
\end{align}
Therefore, minimizing $J$ over $\pi^c_N$ is equivalent to minimizing it over $\bar{\omega}$.

Next, we consider the non-truncated MDP corresponding to the transmission policy $\pi_N^{c^\star}$. Denote the associated occupation measure as $\omega(\tau,(\nu,a))$, where $\tau\in[0,\infty)$. Then, $\omega$ satisfies (\ref{occupation}) with $N=\infty$. Moreover, similar to (\ref{prob_occu}), from (\ref{prob_Q}) we have  
\begin{equation}
\begin{aligned}
&\frac{\omega(\tau,(\nu,a))}{\sum_{\nu_k,a_k\in\{0,1\}}\omega(\tau,(\nu,a))}=\mathbb{P}\big((\nu_k,a_k)\mid\tau\big),~\text{if}~\tau <N-1,
\\
& \frac{\omega(\tau,(\nu,a)=(1,1))}{\sum_{\nu_k,a_k\in\{0,1\}}\omega(\tau,(\nu,a))}=1 \quad\quad\quad\quad\quad\quad~ \text{if}~ \tau\geq N-1.
\end{aligned} \label{omega}
\end{equation}
From (\ref{omega}) and (\ref{exponential}), if $\tau\geq N-1$, then we have
\begin{equation}
\begin{aligned}
&\omega\big(\tau,(\nu,a)\big)=0 \quad \text{if}~(\nu,a)\neq (1,1) \\
&\sum_{\nu,a\in\{0,1\}}\omega\big(\tau+1,(\nu,a)\big) \leq \eta \lambda_a \sum_{\nu,a\in\{0,1\}} \omega\big(\tau,(\nu,a)\big).
\end{aligned}
\end{equation}

Similar to (\ref{trun}), $J^{\eta,\infty}_{(\pi_N^{c^\star},\pi_N^{a^\star})} $ can be expressed in terms of $\omega$ as
\begin{align}
J^{\eta,\infty}_{(\pi_N^{c^\star},\pi_N^{a^\star})} =& \sum_{\tau=0}^\infty\sum_{\nu,a\in\{0,1\}} \omega\big(\tau,(\nu,a)\big)\text{tr}f^\tau\big(\bar{P}\big). \label{longh}
\end{align}	
From (\ref{occupation}) and (\ref{omega}), we have
\begin{equation}
\begin{aligned}
\omega\big(\bar{\tau},(\nu,a)\big)&=\bar{\omega} \big(\bar{\tau},(\nu,a)\big),\quad \bar{\tau}\in[0,N-1), \\
\sum_{i=N-1}^\infty\omega\big(i,(\nu,a)\big)&=\bar{\omega} \big(\bar{\tau},(\nu,a)\big),\quad \bar{\tau}=N-1. 
\end{aligned} \label{equal}
\end{equation}
The condition in  (\ref{condition}) ensures $\sum_{\nu,a\in\{0,1\}}\omega(i,(\nu,a))f^{i}(\bar{P})$  converges exponentially to $0$ with $i$. This guarantees the existence of $\hat{\epsilon}_N$ such that $\lim_{N\rightarrow\infty}\hat{\epsilon}_N=0$ and 
 \begin{equation}
\begin{aligned}
&\Big|\sum_{i=N-1}^\infty\sum_{\nu,a}\omega\big(i,(\nu,a)\big)f^{i}(\bar{P})-\bar{\omega}\big(N-1,(\nu,a)\big)f^{N-1}(\bar{P})\Big|\\
\leq & \Big|\sum_{i=N-1}^\infty\sum_{\nu,a}\omega\big(i,(\nu,a)\big)\big(f^{i}(\bar{P})-f^{N-1}(\bar{P})\big)\Big| \\
\leq &\Big|\sum_{i=N-1}^\infty\sum_{\nu,a}\omega\big(i,(\nu,a)\big)f^{i}(\bar{P})\Big|\leq  \hat{\epsilon}_N,
\end{aligned}\label{error}\end{equation} 
From (\ref{trun}) and (\ref{longh})- (\ref{error}), we can prove that 
 \begin{equation}|J^{\eta,N}_{(\pi_N^{c^\star},\pi_N^{a^\star})}-J^{\eta,\infty}_{(\pi_N^{c^\star},\pi_N^{a^\star})}|\leq \hat{\epsilon}_N. \label{com1}
\end{equation} According to the optimality of $\pi^{c_0}$, we have 
\begin{equation}
J^{\eta,\infty}_{(\pi^{c_0},\pi_N^{a^\star})} \leq J^{\eta,\infty}_{(\pi_N^{c^\star},\pi_N^{a^\star})}.\label{com2}
\end{equation}
(\ref{com1}) and (\ref{com2}) together give that 
\begin{equation}
J^{\eta,\infty}_{(\pi^{c_0},\pi_N^{a^\star})} \leq J^{\eta,N}_{(\pi_N^{c^\star},\pi_N^{a^\star})}+\hat{\epsilon}_N.\label{com3}
\end{equation}
In a similar way, using the boundedness of $J^{\eta,\infty}_{(\pi^{c_0},\pi_N^{a^\star})}$ and the optimality of $\pi_N^{c^\star}$ to the truncated optimization problem, we can prove the existence of $\bar{\epsilon}_N$ such that  $\lim_{N\rightarrow\infty}\bar{\epsilon}_N=0$ and
\begin{equation}
J^{\eta,N}_{(\pi_N^{c^\star},\pi_N^{a^\star})} \leq J^{\eta,\infty}_{(\pi^{c_0},\pi_N^{a^\star})}+\bar{\epsilon}_N.\label{com4}
\end{equation}
Eq. (\ref{com0}), (\ref{com3}), and (\ref{com4}) prove (\ref{min_t}) since any $\epsilon_{\eta,N}$ satisfying that $2\tilde{\epsilon}_\eta+2\max(\hat{\epsilon}_N,\tilde{\epsilon}_N)\leq \epsilon_{\eta,N}$ and $\lim_{\eta\rightarrow 1,N\rightarrow\infty}\epsilon_{\eta,N}=0$ can make (\ref{min_t}) hold. (\ref{min_t}) and (\ref{ep}) give (\ref{ep1}). Eq. (\ref{ep2}) can be proved similarly.
\end{pf}

\bibliographystyle{plain}        
\bibliography{codesign}

\end{document}